\documentclass{llncs}[11pt]

\usepackage[numbers,sort&compress]{natbib}

\makeatletter

\usepackage{etex}

\usepackage{zref-savepos}

\newcounter{mnote}%[page]

\def\xmarginnote{%
  \xymarginnote{\hskip -\marginparsep \hskip -\marginparwidth}}

\def\ymarginnote{%
  \xymarginnote{\hskip\columnwidth \hskip\marginparsep}}

\long\def\xymarginnote#1#2{%
\vadjust{#1%
\smash{\hbox{{%
        \hsize\marginparwidth
        \@parboxrestore
        \@marginparreset
\footnotesize #2}}}}}

\def\mnoteson{%
\gdef\mnote##1{\refstepcounter{mnote}\label{##1}%
  \zsavepos{##1}%
  \ifnum20432158>\number\zposx{##1}%
  \xmarginnote{{\color{blue}\bf $\langle$\arabic{mnote}$\rangle$}}% 
  \else
  \ymarginnote{{\color{blue}\bf $\langle$\arabic{mnote}$\rangle$}}%
  \fi%
}
  }
\gdef\mnotesoff{\gdef\mnote##1{}}
\mnoteson
\mnotesoff

\usepackage{framed}
\usepackage{subcaption}
\usepackage{comment}
\usepackage[svgnames,dvipsnames]{xcolor}
\usepackage[all]{xy}
\usepackage{tikz}
\usetikzlibrary{positioning,matrix,through,calc,arrows,fit,shapes,decorations.pathreplacing,decorations.markings,}

\tikzstyle{block} = [draw,fill=blue!20,minimum size=2em]

\usepackage{qsymbols,amssymb,mathrsfs}
\usepackage{amsmath}
\usepackage[overload]{empheq} % no \intertext and \displaybreak

\let\iftwocolumn\if@twocolumn
\g@addto@macro\@twocolumntrue{\let\iftwocolumn\if@twocolumn}
\g@addto@macro\@twocolumnfalse{\let\iftwocolumn\if@twocolumn}

\mathtoolsset{showonlyrefs=false,showmanualtags}
\newtagform{brackets}[]{(}{)}   % new tags for equations
\usetagform{brackets}
\usepackage[Smaller]{cancel}

\PassOptionsToPackage{breaklinks,letterpaper,hyperindex=true,backref=false,bookmarksnumbered,bookmarksopen,linktocpage,colorlinks,linkcolor=BrickRed,citecolor=OliveGreen,urlcolor=Blue,pdfstartview=FitH}{hyperref}
\usepackage{hyperref}

\usepackage{graphicx,psfrag}
\graphicspath{{figure/}{image/}} % Search path of figures

\usepackage{diagbox} % \backslashbox{}{} for slashed entries
\DeclareMathAlphabet{\mathpzc}{OT1}{pzc}{m}{it}
\usepackage{upgreek} % \upalpha,\upbeta, ...
\usepackage{dsfont}  % \mathds

\def\multi@nostar#1#2{%
  \expandafter\def\csname multi#1\endcsname##1{%
    \if ##1.\let\next=\relax \else
    \def\next{\csname multi#1\endcsname}     
    \expandafter\newcommand\csname #1##1\endcsname{#2}
    \fi\next}}

\def\multi@star#1#2{%
  \expandafter\def\csname #1\endcsname##1{#2}
  \multi@nostar{#1}{#2}
}

\newcommand{\multi}{%
  \@ifstar \multi@star \multi@nostar}

\multi*{rm}{\mathrm{#1}}
\multi*{mc}{\mathcal{#1}}
\multi*{op}{\mathop {\operator@font #1}}
\multi*{ds}{\mathds{#1}}
\multi*{set}{\mathcal{#1}}
\multi*{rsfs}{\mathscr{#1}}
\multi*{pz}{\mathpzc{#1}}
\multi*{M}{\boldsymbol{#1}}
\multi*{R}{\mathsf{#1}}
\multi*{RM}{\M{\R{#1}}}
\multi*{bb}{\mathbb{#1}}
\multi*{td}{\tilde{#1}}
\multi*{tR}{\tilde{\mathsf{#1}}}
\multi*{trM}{\tilde{\M{\R{#1}}}}
\multi*{tset}{\tilde{\mathcal{#1}}}
\multi*{tM}{\tilde{\M{#1}}}
\multi*{baM}{\bar{\M{#1}}}
\multi*{ol}{\overline{#1}}

\multirm  ABCDEFGHIJKLMNOPQRSTUVWXYZabcdefghijklmnopqrstuvwxyz.
\multiol  ABCDEFGHIJKLMNOPQRSTUVWXYZabcdefghijklmnopqrstuvwxyz.
\multitR   ABCDEFGHIJKLMNOPQRSTUVWXYZabcdefghijklmnopqrstuvwxyz.
\multitd   ABCDEFGHIJKLMNOPQRSTUVWXYZabcdefghijklmnopqrstuvwxyz.
\multitset ABCDEFGHIJKLMNOPQRSTUVWXYZabcdefghijklmnopqrstuvwxyz.
\multitM   ABCDEFGHIJKLMNOPQRSTUVWXYZabcdefghijklmnopqrstuvwxyz.
\multibaM   ABCDEFGHIJKLMNOPQRSTUVWXYZabcdefghijklmnopqrstuvwxyz.
\multitrM   ABCDEFGHIJKLMNOPQRSTUVWXYZabcdefghijklmnopqrstuvwxyz.
\multimc   ABCDEFGHIJKLMNOPQRSTUVWXYZabcdefghijklmnopqrstuvwxyz.
\multiop   ABCDEFGHIJKLMNOPQRSTUVWXYZabcdefghijklmnopqrstuvwxyz.
\multids   ABCDEFGHIJKLMNOPQRSTUVWXYZabcdefghijklmnopqrstuvwxyz.
\multiset  ABCDEFGHIJKLMNOPQRSTUVWXYZabcdefghijklmnopqrstuvwxyz.
\multirsfs ABCDEFGHIJKLMNOPQRSTUVWXYZabcdefghijklmnopqrstuvwxyz.
\multipz   ABCDEFGHIJKLMNOPQRSTUVWXYZabcdefghijklmnopqrstuvwxyz.
\multiM    ABCDEFGHIJKLMNOPQRSTUVWXYZabcdefghijklmnopqrstuvwxyz.
\multiR    ABCDEFGHIJKL NO QR TUVWXYZabcd fghijklmnopqrstuvwxyz.
\multibb   ABCDEFGHIJKLMNOPQRSTUVWXYZabcdefghijklmnopqrstuvwxyz.
\multiRM   ABCDEFGHIJKLMNOPQRSTUVWXYZabcdefghijklmnopqrstuvwxyz.

\newcommand{\dotleq}{\buildrel \textstyle  .\over {\smash{\lower
      .2ex\hbox{\ensuremath\leqslant}}\vphantom{=}}}
\newcommand{\dotgeq}{\buildrel \textstyle  .\over {\smash{\lower
      .2ex\hbox{\ensuremath\geqslant}}\vphantom{=}}}

\newcommand{\bM}{\begin{bmatrix}}
\newcommand{\eM}{\end{bmatrix}}
\newcommand{\bSM}{\left[\begin{smallmatrix}}
\newcommand{\eSM}{\end{smallmatrix}\right]}
\renewcommand*\env@matrix[1][*\c@MaxMatrixCols c]{%
  \hskip -\arraycolsep
  \let\@ifnextchar\new@ifnextchar
  \array{#1}}

\newqsymbol{`N}{\mathbb{N}}
\newqsymbol{`R}{\mathbb{R}}
\newqsymbol{`P}{\mathbb{P}}
\newqsymbol{`Z}{\mathbb{Z}}

\newqsymbol{`|}{\mid}
\newqsymbol{`8}{\infty}
\newqsymbol{`1}{\left}
\newqsymbol{`2}{\right}
\newqsymbol{`6}{\partial}
\newqsymbol{`0}{\emptyset}
\newqsymbol{`-}{\leftrightarrow}
\newqsymbol{`<}{\langle}
\newqsymbol{`>}{\rangle}

\DeclarePairedDelimiter\abs{\lvert}{\rvert}

\DeclarePairedDelimiter\Set{\{}{\}}
\newcommand{\imod}[1]{\allowbreak\mkern10mu({\operator@font mod}\,\,#1)}

\newcommand{\threecols}[3]{
\hbox to \textwidth{%
      \normalfont\rlap{\parbox[b]{\textwidth}{\raggedright#1\strut}}%
        \hss\parbox[b]{\textwidth}{\centering#2\strut}\hss
        \llap{\parbox[b]{\textwidth}{\raggedleft#3\strut}}%
    }% hbox 
}

\newcommand{\reason}[2][\relax]{
  \ifthenelse{\equal{#1}{\relax}}{
    \left(\text{#2}\right)
  }{
    \left(\parbox{#1}{\raggedright #2}\right)
  }
}

\newcommand{\utag}[2]{\mathop{#2}\limits^{\text{(#1)}}}
\newcommand{\uref}[1]{(#1)}

\let\SavedDoubleVert\relax
\let\protect\relax
{\catcode`\|=\active
  \xdef\extendvert{\protect\expandafter\noexpand\csname extendvert \endcsname}
  \expandafter\gdef\csname extendvert \endcsname#1{\mskip-5mu \left.%
      \ifx\SavedDoubleVert\relax \let\SavedDoubleVert\|\fi
     \:{\let\|\SetDoubleVert
       \mathcode`\|32768\let|\SetVert
     #1}\:\right.\mskip-5mu}
}
\def\SetVert{\@ifnextchar|{\|\@gobble}% turn || into \|
    {\egroup\;\mid@vertical\;\bgroup}}
\def\SetDoubleVert{\egroup\;\mid@dblvertical\;\bgroup}

\begingroup
 \edef\@tempa{\meaning\middle}
 \edef\@tempb{\string\middle}
\expandafter \endgroup \ifx\@tempa\@tempb
 \def\mid@vertical{\middle|}
 \def\mid@dblvertical{\middle\SavedDoubleVert}
\else
 \def\mid@vertical{\mskip1mu\vrule\mskip1mu}
 \def\mid@dblvertical{\mskip1mu\vrule\mskip2.5mu\vrule\mskip1mu}
\fi

\makeatother

\usepackage{ctable}
\usepackage{fouridx}
\usepackage{framed}
\usetikzlibrary{positioning,matrix}

\usepackage{paralist}
\usepackage{enumerate}

\usepackage[normalem]{ulem}

{\endMakeFramed}

\newenvironment{ybox}{
	\setlength{\FrameSep}{1.5mm}
	\setlength{\FrameRule}{0mm}
  \MakeFramed {\FrameRestore}}%
{\endMakeFramed}

\newenvironment{gbox}{
	\setlength{\FrameSep}{1.5mm}
\setlength{\FrameRule}{0mm}
  \MakeFramed {\FrameRestore}}%
{\endMakeFramed}

{\endMakeFramed}

 {\endMakeFramed}

\usepackage{enumitem}

\usepackage{etoolbox}

\let\theparentequation\theequation
\patchcmd{\theparentequation}{equation}{parentequation}{}{}

\renewenvironment{subequations}[1][]{%              optional argument: label-name for (first) parent equation
	\refstepcounter{equation}%
	\setcounter{parentequation}{\value{equation}}%    parentequation = equation
	\setcounter{equation}{0}%                         (sub)equation  = 0
	\def\theequation{\theparentequation\alph{equation}}% 
	\let\parentlabel\label%                           Evade sanitation performed by amsmath
	\ifx\\#1\\\relax\else\label{#1}\fi%               #1 given: \label{#1}, otherwise: nothing
	\ignorespaces
}{%
	\setcounter{equation}{\value{parentequation}}%    equation = subequation
	\ignorespacesafterend
}

\newcommand*{\nextParentEquation}[1][]{%            optional argument: label-name for (first) parent equation
	\refstepcounter{parentequation}%                  parentequation++
	\setcounter{equation}{0}%                         equation = 0
	\ifx\\#1\\\relax\else\parentlabel{#1}\fi%         #1 given: \label{#1}, otherwise: nothing
}

\newcommand{\RCO}{R_{\op{CO}}}
\newcommand{\tRCO}{\tilde{R}_{\op{CO}}}
\newcommand{\RS}{R_{\op{S}}}

\newcommand{\CS}{C_{\op{S}}}
\newcommand{\tCS}{\tilde{C}_{\opS}}

\usepackage[outline]{contour}
\contourlength{1.2pt}

\title{Compressed Secret Key Agreement}
\subtitle{Maximizing Multivariate Mutual Information Per Bit}

\author{Chung Chan}
\institute{Institute of Network Coding,\\ The Chinese University of Hong Kong, Hong Kong.}
	
\begin{document}
	\maketitle
	\begin{abstract}
The multiterminal secret key agreement problem by public discussion is formulated with an additional source compression step where, prior to the public discussion phase, users independently compress their private sources to filter out strongly correlated components for generating a common secret key. The objective is to maximize the achievable key rate as a function of the joint entropy of the compressed sources. Since the maximum achievable key rate captures the total amount of information mutual to the compressed sources, an optimal compression scheme essentially maximizes the multivariate mutual information per bit of randomness of the private sources, and can therefore be viewed more generally as a dimension reduction technique. Single-letter lower and upper bounds on the maximum achievable key rate are derived for the general source model, and an explicit polynomial-time computable formula is obtained for the pairwise independent network model. In particular, the converse results and the upper bounds are obtained from those of the related secret key agreement problem with rate-limited discussion. A precise duality is shown for the two-user case with one-way discussion, and such duality is extended to obtain the desired converse results in the multi-user case. In addition to posing new challenges in information processing and dimension reduction, the compressed secret key agreement problem helps shed new light on resolving the difficult problem of secret key agreement with rate-limited discussion, by offering a more structured achieving scheme and some simpler conjectures to prove.
	\end{abstract}

\section{Introduction}
\label{sec:introduction}
%\input{intro}

% why source reduction?
% Not everything about the source is needed for SKA. There can be redundancy.
% Sometimes one need to generate a large amount of key. Only a reasonably small key. Reducing the source save resources and simplifies the later steps.
% It actually offer good upper bounds for studying the public discussion required for secret key agreement.
% viewed as an extension of the information bottleneck problem... 
% we also reveal much open problems

In Information-theoretic security, the secret key agreement problem by public discussion is the
problem where 
a group of users discuss in public to generate a common secret key that is
independent of their discussion. 
The problem was first formulated by Maurer~\cite{maurer93}, Ahlswede
and Csisz\'ar~\cite{ahlswede93} under a private source model involving two users who observe
some correlated private sources. Rather surprisingly, public
discussion was shown to be useful in generating the secret key, i.e., it strictly increases the
maximum achievable key rate called the \emph{secrecy capacity}. Such phenomenon was also discovered in \cite{bennett1988privacy} in a different formulation. Furthermore, the secrecy capacity was given an information-theoretically appealing characterization--- it is equal to Shannon's mutual information~\cite{shannon48} between the two private sources, assuming the wiretapper can listen to the entire public discussion but not observe any other side information of the private sources. It was also shown that the capacity can be achieved by one-way public discussion, i.e., with only one of the users discusses in public. 

As a simple illustration, let $\RX_0$, $\RX_1$ and $\RJ$ be three uniformly random independent bits, and suppose user~$1$ observes $\RZ_1:=(\RX_0,\RX_1)$ privately while user~$2$ observes $\RZ_2:=(\RX_{\RJ},\RJ)$, where $\RX_{\RJ}=\RX_0$ when $\RJ=0$ but $\RX_{\RJ}=\RX_1$ when $\RJ=1$.
If user~$2$ reveals $\RJ$ in public, then user~$1$ can recover $\RX_{\RJ}$ and therefore $\RZ_2$. Furthermore, since $\RX_{\RJ}$ is independent of $\RJ$, it can serve as a secret key bit that is recoverable by both users but remains perfectly secret to a wiretapper who observes only the public message $\RJ$. This scheme achieves the secrecy capacity equal to the mutual information $I(\RZ_1\wedge \RZ_2)=1$ roughly because user~$2$ reveals $H(\RZ_2|\RZ_1)=1$~bit in public so there is $H(\RZ_2)-H(\RZ_2|\RZ_1)=I(\RZ_1\wedge \RZ_2)$~bits of randomness left for the secret key. However, if no public discussion is allowed, it follows from the work of G\'ac and K\"orner~\cite{gac72} that no common secret key bit can be extracted from the sources. In particular, $\RX_{\RJ}$ cannot be used as a secret key because user~$1$ does not know whether $\RX_{\RJ}$ is $\RX_0$ or $\RX_1$. $\RX_0$ and also $\RX_1$ cannot be used as a secret key either because they may not be observed by user~$2$ when $\RJ=1$ and $\RJ=0$ respectively. It can be seen that, while the private sources are clearly statistical dependent, public discussion is needed to consolidate the mutual information of the sources into a common secret key.

The secret key agreement formulation was subsequently extended to the multi-user case by Csisz\'ar and Narayan~\cite{csiszar04}. Some users are also allowed to act as \emph{helpers} who can participate in the public discussion but need not share the secret key. The designated set of users who need to share the secret key are referred to as the \emph{active users}. Different from the two-user case, one-way discussion may not achieve the secrecy capacity when there are more than two users. Instead, an \emph{omniscience strategy} was considered in \cite{csiszar04} where the users first communicate minimally in public until omniscience, i.e., the users discuss in public at the smallest total rate until every active user can recover all the private sources. The scheme was shown to achieve the secrecy capacity in the case when the wiretapper only listens to the public discussion. This assumes, however, that the public discussion is lossless and unlimited in rate, and the sources take values from finite alphabet sets. If the sources were continuous or if the public discussion were limited to a certain rate, it may be impossible to attain omniscience. 

This work is motivated by the search of a better alternative to the omniscience strategy for multiterminal secret key agreement. A prior work of Csisz\'ar and Narayan~\cite{csiszar00} considered secret key agreement under rate-limited public discussion. The model involves two
users and a helper observing correlated discrete memoryless sources. The public discussion by the users is conducted in a particular order and direction. While the region of achievable secret key rate and discussion rates remains unknown, single-letter characterizations involving two auxiliary random variables were given for many special cases, including the two-user case with two rounds of interactive public discussion, where each user speaks once in sequence, with the last public message possibly depending on the first. By further restricting to one-way public discussion, the characterization involves only one auxiliary random variable and was extended to continuous sources by Watanabe and Oohama in~\cite{watanabe10}, where they also gave an explicit characterization without any auxiliary random variable for scalar Gaussian sources in \cite{watanabe10}. For vector Gaussian sources, the characterization by the same authors in \cite{watanabe11} involving some matrix optimization was further improved in \cite{liu16} to a more explicit formula. However, if the discussion is allowed to be two-way and interactive, Tyagi~\cite{tyagi13} showed with a concrete two-user example that the minimum total discussion rate required, called the \emph{communication complexity}, can be strictly reduced. Using the technique of Kaspi~\cite{kaspi85}, multi-letter characterizations were given in \cite{tyagi13} for the communication complexity and, similarly, by Liu et al.\ in \cite{LCV16} for the region of achievable secret key rate. \cite{LCV16} further simplified the characterization using the idea of convex envelope using the technique by Ma et al~\cite{ma12}. While these characterizations provide many new insights and properties, they are not considered computable, compared to the usual single-letter and explicit characterizations. Further extension to the multi-user case also appears difficult, as the converse can be seen to rely on the Csisz\'ar sum identity~\cite[Lemma~4.1]{ahlswede93}, which does not appear to extend beyond the two-user case.

Nevertheless, partial solutions under more restrictive public discussion constraints were possible. By simplifying the problem to the right extent, new results were discovered in the multi-user case, which has led to the formulation in this work. For instance, Gohari and Anantharam~\cite{amin10a} characterized the secrecy capacity in the multi-user case under the simpler vocality constraint where some users have to remain silent throughout the public discussion. Using this result, simple necessary and sufficient conditions can be derived as to whether a user can remain silent without diminishing the maximum achievable key rate~\cite{mukherjee14,zhang15,chan16so}. This is a simpler result than characterizing the achievable rate region because it does not say how much discussion is required if a user must discuss. Another line of work~\cite{courtade16,mukherjee16,MKS16,chan16itw} follows \cite{tyagi13} to characterize the communication complexity but in the multi-user case. Courtade and Halford~\cite{courtade16} characterized the communication complexity under a special non-asymptotic hypergraphical source model with linear discussion. 
\cite{MKS16} obtained a multi-letter lower bound on the communication complexity for the asymptotic general source model. It also gave a precise and simple condition under which the omniscience strategy for secret key agreement is optimal for a special source model called the \emph{pairwise independent network (PIN)}~\cite{nitinawarat10}, which is a special hypergraphical source model~\cite{chan10md}. \cite{chan16itw,chan17oo} further derived some single-letter and more easily computable explicit lower bounds, from which one can also obtain conditions for the omniscience strategy to be optimal under the hypergraphical source model, which covers the PIN model as a special case. \cite{chan17cska} considered the more general problem of characterizing the multiterminal secrecy capacity under rate-limited public discussion. In particular, an objective of \cite{chan17cska} is to characterize the \emph{constrained secrecy capacity} defined as the maximum achievable key rate as a function of the total discussion rate. This covers the communication complexity as a special case when further increase in the public discussion rate does not increase the secrecy capacity. While only single-letter bounds were derived for the general source model, a surprisingly simple explicit formula was derived for the PIN model~\cite{chan17cska}. The optimal scheme in \cite{chan17cska} follows the tree-packing protocol in \cite{nitinawarat-ye10}. It turns out to belong to the more general approach of decremental secret key
agreement in~\cite{chan16isit,chan17idska} inspired by the achieving scheme in \cite{courtade16} and the notion of excess edge in \cite{chan10md}. More precisely, the omniscience strategy is applied after some excess or less useful edge random
variables are removed (decremented) from the source. Since the entropy of the decremented source is smaller, the discussion required to attain omniscience of the decremented source is also smaller. Such decremental secret key agreement approach applies to hypergraphical sources more generally, and it results in one of the best upper bounds in~\cite{mukherjee16} for communication complexity. However, for more general source models that are not necessarily hypergraphical, the approach does not directly apply.

The objective of this work is to formalize and extend the idea of decremental secret key agreement beyond the hypergraphical source model. More precisely, the secret key agreement problem is considered with an additional source compression step before public discussion
where each user independently compresses their private source component to filter away less correlated randomness that does not contribute much to the achievable secret key rate.
The compression is such that the entropy rate of the compressed sources is reduced to under certain specified level. In particular, the edge removal process in decremental secret key agreement can be viewed as a special case of source compression, and the more general problem will be referred to as \emph{compressed secrecy key agreement}. The objective is to characterize the
achievable secret key rate maximized over all valid
compression schemes. For simplicity, this work will focus on the case without helpers, i.e., when all users
are active and want to share a common secret key. A closely related formulation is by Nitinawarat and Narayan~\cite{nitinawarat12}, which characterized the maximum achievable key rate for the two-user case under the scalar gaussian source model where one of the user is required to quantize the source to within a given rate. \cite{vatedka16} also extended the formulation and techniques in \cite{nitinawarat12} to the multi-user case where every user can quantize their sources individually to a certain rate. The compression considered in this work is more general than quantizations for gaussian sources, and the new results are meaningful beyond continuous sources.

The compressed secret key agreement problem is also motivated by the study of multivariate mutual information (MMI)~\cite{chan15mi}, i.e., an
extension of Shannon's mutual information to the multivariate case involving possibly more than
two random variables. The unconstrained secrecy capacity in the no-helper case has been viewed as a measure of mutual information in
\cite{chan2008tightness,chan15mi}, not only because of its mathematically appealing interpretations such as the residual independence relation and data processing inequalities in \cite{chan15mi}, but also because of its operational significance in undirected network coding~\cite{chan11isit,chan12ud}, data clustering~\cite{chan16cluster} and feature selection~\cite{chan16allerton} (cf.\ \cite{csiszar2008axiomatic}). The optimal source compression scheme that achieves the compressed secrecy capacity can be viewed more generally as an optimal dimension reduction procedure that maximizes the MMI per bit of randomness, which is an extension of the information bottleneck problem~\cite{tishby00} to the multivariate case. However, different from the multivariate extension in~\cite{friedman01}, the MMI is used instead of Watanabe's total correlation~\cite{watanabe60}, and so it captures only the information mutual to all the random variables rather than the information mutual to any subsets of the random variables. Furthermore, the compression is on each random variable rather than subsets of random variables.

%	%\section{Connection to Multivariate Mutual Information}
%	It is worthwhile to mention a motivation of the problem beyond information-theoretic security.
%Indeed, compressed secret key agreement can be viewed more generally as a dimension reduction problem. There was an upper bound in \cite{csiszar04} for the secrecy capacity in terms of the divergence from the joint distribution to a product of marginal distributions, much like Shannon's mutual information. Although the bound was loose in general, we found in~\cite{chan2008tightness,chan10md} that it is tight in the no-helper case when all users are active. The tightness of the divergence bound inspired \cite{chan2008tightness} to propose the secrecy capacity in the no-helper case as a measure of mutual information for multiple random variables, referred to as the multivariate mutual information (MMI)~in \cite{chan15mi}. Subsequently, the secrecy capacity in the no-helper case was also referred to as the shared information by Narayan and Tyagi~\cite{CIT-072}. The MMI was shown to be applicable to some machine learning problems such as feature selection~\cite{chan16allerton} and data clustering~\cite{chan16cluster}. Now, with the secrecy capacity regarded as a general measure of mutual information among a set of random variables, what an optimal compressed secret key agreement scheme is trying to do is to maximize the MMI for each bit of randomness. Having a larger MMI per bit means one can perform inference problems on the random variables more accurately within a given level model complexity.

The paper is organized as follows. The problem of compressed secret key agreement is formulated in Section~\ref{sec:problem}. Preliminary results of secret key agreement are given in Section~\ref{sec:prelim}. The main results are motivated in Section~\ref{sec:ML} and presented in Section~\ref{sec:results}, followed by the conclusion and some discussions on potential extensions in Section~\ref{sec:conclusion}.

\section{Problem Formulation}
\label{sec:problem}

Similar to the multiterminal secret key agreement problem~\cite{csiszar04} without helpers or
wiretapper's side information, the setting of the problem involves a finite set $V$ of $\abs {V}>1$ users, and a discrete
memoryless multiple source
\begin{align*}
\RZ_V & := (\RZ_i|i\in V)\sim P_{\RZ_V} \text{ taking values from }\\
Z_V & := \prod\nolimits_{i\in V} Z_i \kern1em \text{(not necessarily finite).}
\end{align*}
N.b., letters in sans serif font are used for random variables and the corresponding capital letters in the usual math italic font denote the alphabet sets. $P_{\RZ_V}$ denotes the joint distribution of $\RZ_i$'s. 

A secret key agreement protocol with source compression can be broken into the following phases:
%\begin{lbox}
\setlength{\plparsep}{1em}
\begin{compactdesc}
	\item[Private observation:] Each user $i\in V$ observes an $n$-sequence 
	$$
	\RZ_i^n:=(\RZ_{it}|t\in[n])=(\RZ_{i1},\RZ_{i2},\ldots,\RZ_{in})
	$$
	i.i.d. generated from the source $\RZ_i$ for some block length $n$. N.b., for convenience, $[n]$ denotes the set of positive intergers up to $n$, i.e, $\Set {1,\dots,n}$.
	
	\item[Private randomization:] Each user $i\in V$ generates a random variable $\RU_i$ independent of the private source, i.e.,
	\begin{align}
	\label{eq:U}
	H(\RU_V|\RZ_V)=\sum_{i\in V}H(\RU_i).
	\end{align}
	
	\item[Source compression:] 	Each user $i\in V$ computes
	\begin{align}
	\tRZ_i &= `z_i(\RU_i,\RZ_i^n) \label{eq:tZ}
	\end{align}
	for some function $`z_i$ that maps to a finite set. $\tRZ_V$ is referred to as the compressed source. 
	
	\item[Public discussion:]  Using a public authenticated noiseless channel, a user $i_t\in V$ is chosen in round~$t \in [\ell]$ to broadcast a message
	\begin{subequations}
		\label{eq:discussion}
		\begin{alignat}{2}
		\tRF_t &:=\tdf_t(\tRZ_{i_t},\tRF^{t-1})&\kern1em&\text {where} \label{eq:Fit}
		\end{alignat}
		$\ell$ is a positive integer denoting the number of rounds and
		$\tRF^{t-1}$ denotes all the messages broadcast in the previous rounds. If the dependency on $\tRF^{t-1}$ is dropped, the discussion is said to be \emph{non-interactive}. The discussion is said to be one-way (from user $i$) if $\ell=1$ (and $i_1=1$). For convenience,
		\begin{align}
		\RF_i & := (\tRF_t|t\in[\ell],i_t=i)\label{eq:Fi}\\
		\RF & := \tRF^{\ell} = \RF_V \label{eq:F}
		\end{align}
		denote the aggregate message from user $i\in V$ and the aggregation of the messages from all users respectively.
	\end{subequations}
	
	\item[Key generation:] A random variable $\RK$, called the secret key, is required to satisfy the recoverability constraint that
	\begin{equation}
	\lim_{n\to\infty}\text{Pr}(\exists \; i\in V, \RK\neq\theta_i(\tRZ_i,\RF))=0, \label{eq:recover}
	\end{equation}
	for some function $\theta_i$, and the secrecy constraint that
	\begin{equation}
	\lim_{n\to\infty}\frac{1}{n}`1[\log\abs{K}-H(\RK|\RF)`2]=0, \label{eq:secrecy}
	\end{equation}
	where $K$ denotes the finite alphabet set of possible key values.
\end{compactdesc}
N.b., unlike \cite{tyagi13}, non-interactive discussion is considered different from one-way discussion in the two-user case since both users are allowed to discuss even though their messages cannot depend on each other. Different from \cite{csiszar08}, there is an additional source compression phase, after which the protocol can only depend on the origninal sources through the compressed sources. 

The objective is to characterize the maximum achievable secret key rate for a continuum of different levels of source compression:
\begin{definition}
	\label{def:tCS}
	The \emph{compressed secrecy capacity} with a joint entropy limit $`a\geq 0$ is defined as 
	\begin{align}
	\tCS(`a) &:= \sup \liminf_{n\to\infty}\frac{1}{n}\log\abs{K}
	\label{eq:tCS}
	\end{align}
	where the supremum is over all possible compressed secret key agreement schemes satisfying
	\begin{align}
	\limsup_{n\to `8} \frac1n H(\tRZ_V) - `a\leq 0. \label{eq:`a}
	\end{align}
	This constraint limits the joint entropy rate of the compressed source.
\end{definition}
N.b., instead of the joint entropy limit, one may also consider entropy limits on some subset $B \subseteq V$ that
\begin{align}
\limsup_{n\to\infty} \frac1n H(\tRZ_B) - `a \leq 0. \label{eq:`a_B}
\end{align}
If multiple entropy limits are imposed, $\tCS$ will be a higher-dimensional surface instead of a one-dimensional curve. 
For example, in the two-user case under the scalar gaussian source model, \cite{nitinawarat12} considered the entropy limit only on one of the users.
In the multi-user case under the gaussian markov tree model, \cite{vatedka16} considered the symmetric case where the entropy limit is imposed on every user.

For simplicity, however, the joint entropy constraint~\eqref{eq:`a} will be the primary focus in this work. It will be shown that
$\tCS(`a)$ is closely related to the \emph{constrained secrecy capacity} $\CS(R)$ defined as~\cite{chan17cska}
\begin{align}
\CS(R) &:= \sup \liminf_{n\to\infty}\frac{1}{n}\log\abs{K} \kern1em\text {for $R\geq 0$,}\label{eq:CSR}
\end{align}
with $\tRZ_i := (\RU_i,\RZ_i^n)$ instead of \eqref{eq:tZ}, i.e., without compression, and the entropy limit~\eqref{eq:`a} replaced by the constraint on the total discussion rate
\begin{align}
R \geq \limsup_{n\to\infty} \frac{1}{n}\log\abs{F}=\limsup_{n\to\infty} \frac{1}{n}\sum_{i \in V}\log\abs{F_i}. \label{eq:R}
\end{align}
N.b., it follows directly from the result of \cite{csiszar04} that $\tCS(`a)$ remains unchanged whether the discussion is interactive or not. Indeed, the relation between $\tCS(`a)$ and $\CS(R)$ to be shown in this work will not be affected either. Therefore, for notational simplicity, $\CS(R)$ may refer to the case with or without interaction, even though $\CS (R)$ may be smaller with non-interactive discussion.

It is easy to show that $\CS(R)$ is continuous, non-decreasing and concave in $R$~\cite[Proposition~3.1]{chan17cska}. As $R$ goes to $`8$, the secrecy capacity
\begin{align}
\CS(`8) &:=\liminf_{R\to\infty}\CS(R) \label{eq:C`8}
\end{align}
is the usual \emph{unconstrained secrecy capacity} defined in~\cite{csiszar04} without the discussion rate constraint~\eqref{eq:R}. The smallest discussion rate that achieves the unconstrained secrecy capacity is the \emph{communication complexity} denoted by
\begin{align}
\RS:=\inf \Set {R\geq 0\mid \CS(R)=\CS(`8)}.\label{eq:RS}
\end{align}

Similar to $\CS(R)$, the following basic properties can be shown for $\tCS(`a)$:
\begin{proposition}
	\label{pro:basic}
	$\tCS(`a)$ is continuous, non-decreasing and concave in $`a\geq 0$. Furthermore,
	\begin{align}
	\CS(`8)=\liminf_{`a\to\infty}\tCS(`a), \label{eq:C`8:`a}
	\end{align}
	achieving the unconstrained secrecy capacity in the limit.
\end{proposition}
% What is the smallest value when $I(\RZ_V)$ is reached? Define. Also, characterizing the case when alpha is close to but larger than JGK.

\begin{proof}
	Continuity, monotonicity and \eqref{eq:C`8:`a} follow directly from the definition of $\tCS(`a)$. Concavity follows from the usual time-sharing argument, i.e., for any $`l\in [0,1]$, $`a',`a''>0$, a secret key rate of $`l\tCS(`a')+(1-`l)\tCS(1-`a')$ is achievable with the entropy limit $`a:=`l`a'+(1-`l)`a''$ by applying the optimal scheme that achieves $\tCS(`a')$ for the first $n':=\lfloor `ln\rfloor$ samples of $\RZ_V^n$ and applying the optimal scheme that achieves $\tCS(`a'')$ for the remaining $n'':=n-n'$ samples.
\end{proof}

Because of \eqref{eq:C`8:`a}, a quantity playing the same role of $\RS$ for $\CS$ can be defined for $\tCS(`a)$ as follows.

\begin{definition}
	\label{def:`aS}
	The smallest entropy limit that achieves the unconstrained secrecy capacity is defined as
	\begin{align}
	`a_{\opS} := \inf \Set {`a \mid \tCS(`a)=\CS(`8)} \label{eq:`aS}
	\end{align}
	and referred to as the \emph{minimum admissible joint entropy}.
\end{definition}
One may also consider both the entropy limit~\eqref{eq:`a} and discussion rate constraint~\eqref{eq:R} simultaneously, and define the secrecy capacity as a function of $`a$ and $R$. For simplicity, however, we will not consider this case but, instead, focus on the relationship between $\tCS(`a)$ and $\CS(R)$.

The following example illustrates the problem formulation. It will be revisited at the end of Section~\ref{sec:results} (Example~\ref{eg:mot:2}) to illustrate the main results.

\begin{example}
	\label{eg:mot}
	Consider $V:=\Set {1,2,3}$ and
	\begin{align}
	\RZ_1:=(\RX_a,\RX_b),\kern1em \RZ_2:=(\RX_a,\RX_b,\RX_c),\kern1em \text {and}\kern1em \RZ_3:=(\RX_a,\RX_c),\label{eq:mot}
	\end{align}
	where $\RX_a, \RX_b$ and $\RX_c$ are uniformly random and independent bits. It is easy to argue that
	\begin{subequations}
		\begin{align}
		\tCS(`a)\geq `a \kern1em \text {for $`a\in [0,1]$}.
		\label{eq:mot:tCS1}
		\end{align}
		To see this, notice that $\RX_a$ is observed by every user. Any choice of $\RK=`q(\RX_a^n)$ can therefore be recovered by every user without any discussion, satisfying the recoverability constraint~\eqref{eq:recover} trivially. Since there is no public discussion required, the secrecy constraint~\eqref{eq:secrecy} also holds immediately by taking a portion of the bits from $\RX_a^n$ to be the key bits in $\RK$. Finally, setting $\tRZ_i=`z_i(\RX_a^n)=`q(\RX_a^n)$ for all $i\in V$ ensures $H(\tRZ_V)\leq H(\RK)$, satisfying the entropy limit~\eqref{eq:`a} with $`a$ equal to the key rate. Hence, $\tCS(`a)\geq `a$ as desired. Indeed, we will show (by Proposition~\ref{pro:GK}) that the reverse inequality holds in general, and so we have equality for $`a\in [0,1]$ for this example. %Furthermore, equality does not hold for $`a>1$.
		
		For $`a=H(\RZ_V)=H(\RX_a,\RX_b,\RX_c)=3$, every user can simply retain
		their source without compression, i.e., with $\tRZ_i=\RZ_i$ for $i\in V$ while satisfying the entropy limit~\eqref{eq:`a}. Now, with $\RK=(\RX_a^n,\RX_b^n)$ and $\RF=\RF_2=\RX_b^n\oplus \RX_c^n$ where $\oplus$ is the elementwise XOR, it can be shown that both the recoverability~\eqref{eq:recover} and secrecy~\eqref{eq:secrecy} constraints hold.
		This is because user $3$ can recover $\RX_b$ from the XOR $\RX_b\oplus \RX_c$ with the side information $\RX_c$. Furthermore, the XOR bit is independent of $(\RX_a,\RX_b)$ and therefore does not leak any information about the key bits. With this scheme, $\tCS(3)\geq 2$. By the usual time-sharing argument,
		\begin{align}
		\tCS(`a) &\geq
		\begin{cases}
		\frac{1+`a}2 & \text {for $`a\in [1,3]$}\\
		2 & \text {for $`a\geq 3$.}
		\end{cases}\label{eq:mot:tCS2}
		\end{align}
		Indeed, the reverse inequality can be argued using one of the main results (Theorem~\ref{thm:duality}) and so the minimum admissible joint entropy will turn out to be $`a_{\opS}=3$.
	\end{subequations}
\end{example}

%While a single-letter characterization is given for unconstrained secrecy capacity $\CS(`8)$, no single-letter characterization is known for both $\CS(R)$ and $\tCS(`a)$. $\CS(R)$ has been considered before~\cite{} and the question is how one can efficiently discuss to extract the longest shared secret key. The formulation of $\tCS(`a)$ is new, and the question is how one should process the source to capture as much information as possible to extract the longest secret key possible. While these two questions are seemingly different, there is a simple relation between the two quantities.

% relate to previous formulations
% how to choose the processing appropriately

\section{Preliminaries}
\label{sec:prelim}

In this section, a brief summary of related results for the secrecy capacity and communication complexity will be given. The results for the two-user case will be introduced first, followed by the more general results for the multi-user case, and the stronger results for the special hypergraphical source model. An example will also be given at the end to illustrate some of the results.

\subsection{Two-user case}
\label{sec:2user}

As mentioned in the introduction, no single-letter characterization is known for $\CS(R)$ and $\tCS(`a)$ even in the two-user case where $V:=\Set {1,2}$. Furthermore, while multi-letter characterizations for $R_{\opS}$ and $\CS(R)$ were given in \cite{tyagi13} and \cite{LCV16} respectively in the two-user case under interactive discussion, no such multi-letter characterization is known for the case with non-interactive discussion. Nevertheless, if one-way discussion from user~$1$ is considered, then the result of \cite[Theorem~2.4]{csiszar00} and its extension~\cite{watanabe10} to continuous sources gave the following characterization of $\CS(R)$:
\begin{subequations}
	\label{eq:CS1}
	\begin{align}
	C_{\op {S,1}}(R) &:= \sup I(\RZ'_1\wedge \RZ_2) \kern1em \text {where}\label{eq:CS1:1}\\
	& I(\RZ'_1\wedge \RZ_1) - I(\RZ'_1\wedge \RZ_2) \leq R \label{eq:CS1:2}\\
	& I(\RZ'_1\wedge \RZ_2|\RZ_1)=0.
	\label{eq:CS1:3}
	\end{align}
\end{subequations}
The last constraint~\eqref{eq:CS1:3} corresponds to the Markov chain $\RZ'_1"-"\RZ_1"-"\RZ_2$ and so the supremum is taken over the choices of the conditional distribution $P_{\RZ'_1|\RZ_1}=P_{\RZ'_1|\RZ_1,\RZ_2}$. Using the double Markov property as in \cite{tyagi13}, it follows that $\CS(0)$ can be characterized more explicitly by the G\'acs--K\"orner common information 
\begin{align}
J_{\op {GK}}(\RZ_1\wedge \RZ_2) := \sup \Set {H(\RU)\mid H(\RU|\RZ_1)=H(\RU|\RZ_2)=0}\label{eq:JGK:2}
\end{align}
where $\RU$ is a discrete random variable. If \eqref{eq:JGK:2} is finite, a unique optimal solution $\RU$ exists and is called the \emph{maximum common function} of $\RZ_1$ and $\RZ_2$ because any common function of $\RZ_1$ and $\RZ_2$ must be a function of $\RU$. The communication complexity also has a more explicit characterization~\cite[(44)]{tyagi13}
\begin{align}
R_{\op {S,1}} &= J_{\opW,1}(\RZ_1\wedge \RZ_2)-I(\RZ_1\wedge \RZ_2) \quad\text {where}\label{eq:RS1}\\ 
&J_{\opW,1}(\RZ_1\wedge \RZ_2) := \inf \Set {H(\RW) \mid H(\RW|\RZ_1)=0, I(\RZ_1\wedge \RZ_2|\RW)=0 }\label{eq:JW1}
\end{align}
and $\RW$ is a discrete random variable.
If $J_{\opW,1}(\RZ_1\wedge \RZ_2)$ is finite, a unique optimal solution $\RW$ exists and is called the \emph{minimum sufficient statistics} of $\RZ_1$ for $\RZ_2$ since $\RZ_2$ can only depend on $\RZ_1$ through $\RW$.

In Section~\ref{sec:ML}, the expression $C_{\op {S,1}}(R)$ will be related to the compressed secret key agreement restricted to the two-user case when the entropy limit is imposed only on user~$1$. This duality relationship in the two-user case will serve as the motivation of the main results for the multi-user case. 
Indeed, the desired characterization of $\tCS(`a)$ for the two-user case has appeared in \cite[Lemma~4.1]{nitinawarat12} for the scalar gaussian source model:
\begin{subequations}
	\label{eq:tCS1}
	\begin{align}
	\tdC_{\op {S,1}}(`a) &:= \sup I(\RZ'_1\wedge \RZ_2) \kern1em \text {where}\label{eq:CS1:1}\\
	& I(\RZ'_1\wedge \RZ_1)  \leq `a \label{eq:tCS1:2}\\
	& I(\RZ'_1\wedge \RZ_2|\RZ_1)=0.
	\label{eq:tCS1:3}
	\end{align}
\end{subequations}
For the general source model, the expression \eqref{eq:tCS1} has also appeared before with other information-theoretic interpretations as mentioned in~\cite{erkip98}. The lagrangian dual of \eqref{eq:tCS1}, in particular, reduces to the dimension reduction technique called the information bottleneck method in~\cite{tishby00}, where $\RZ_1$ is an observable used to predict the target $\RZ_2$, and $\RZ'_1$ is a feature of $\RZ_1$ that captures as much mutual information with the target variable as possible per bit of mutual information with the observable. Interestingly, the principal of the information bottleneck method was also proposed in ~\cite{tishby15,shwartz-ziv17} as a way to understand   deep learning, since the best prediction of $\RZ_2$ from $\RZ_1$ is nothing but a particular feature of $\RZ_1$ sharing a lot of mutual information with $\RZ_2$. 

\subsection{General source with finite alphabet set}

Consider the multi-user case where $\abs {V}\geq 2$. If $\RZ_V$ takes values from a finite set, then the unconstrained secrecy capacity was shown in~\cite{csiszar04} to be achievable via \emph{communication for omniscience (CO)} and equal to 
\begin{align}
\CS(`8) &= H(\RZ_V) - \RCO, \label{eq:CSRCO}
\end{align}
where $\RCO$ is the smallest rate of CO~\cite{csiszar04} characterized by the linear program
\begin{subequations}
	\label{eq:RCO}
	\begin{align}
	\RCO &= \min_{r_V} r(V)\kern1em \text {such that}\\
	r(B) &\geq H(\RZ_B|\RZ_{V`/B})\kern1em \forall B\subsetneq V,
	\end{align}
\end{subequations}
where $r(B)$ denotes the sum $\sum_{i \in B} r_i$. Further, $R_{\op {CO}}$ can be achieved by non-interactive discussion.
It follows that
\begin{subequations}
	\begin{alignat}{2}
	\RS&\leq \RCO, &\kern1em& \text {or equivalently}\\
	\CS(R)&=\CS(`8) && \text{ $R\geq \RCO$.}
	\label{eq:CSR:RCO}
	\end{alignat}
\end{subequations} 
It was also pointed out in~\cite{csiszar04} that private randomization does not increase $\CS(`8)$. Hence, if $Z_V$ is finite, we have 
\begin{align}
`a_{\opS}\leq H(\RZ_V) \label{eq:`aS:UB}
\end{align}
because $\CS(`8)$ can be achieved with $\tRZ_i=\RZ_i$.
While it seems plausible that randomization does not decrease $\RS$ nor increase $\CS(R)$ for any $R\geq 0$, a rigorous proof remains elusive. Similarly, it appears plausible that neither $`a_{\opS}$ nor $\tCS(`a)$ are affected by randomization but, again, no proof is known yet. %Also, in the case when $Z_V$ is not finite, $\CS(R)$ may be unbounded even for a finite $R$ but $\tCS(`a)$ is always bounded given a finite entropy limit $`a$ (by Proposition~\ref{pro:GK}). It is however unclear whether the unconstrained secrecy capacity, even if finite, can always be reached at a finite discussion rate or entropy limit. In particular, if $\RZ_1$ and $\RZ_2$ are jointly gaussian with correlation coefficient $`r\in (0,1)$, then $\CS(`8)=-\frac12 \log (1-`r^2)< `8$ but it is unclear how to achieve it with finite amount of public discussion or after compressing the source to a finite alphabet set. It seems likely that $\RS$ and $`a_{\opS}$ are both unbounded.

An alternative characterization of $\CS(`8)$ was established in \cite{chan2008tightness,chan10md} by showing that the divergence bound in \cite{csiszar04} is tight in the case without helpers. More precisely, with $\Pi'(V)$ defined as the set of partitions of $V$ into at least two non-empty disjoint sets, then
\begin{subequations}
	\label{eq:mmi}
	\begin{align}
	\CS(`8) = I(\RZ_V)&:=\min_{\mcP\in \Pi'(V)} I_{\mcP}(\RZ_V), \kern1em \text{where} \label{eq:I}\\
	I_{\mcP}(\RZ_V)&:=\frac{1}{|\mcP|-1} D`1(\extendvert{P_{\RZ_V}\|\prod_{C\in \mcP} P_{\RZ_C}}`2) \label{eq:IP}\\ &=\frac{1}{|\mcP|-1}`1[\sum\nolimits_{C\in \mcP}H(\RZ_C) - H(\RZ_V) `2].\kern-2em
	\notag
	\end{align}
\end{subequations}
In the bivariate case when $V=\Set{1,2}$, $I(\RZ_V)$ reduces to Shannon's mutual information $I(\RZ_1\wedge \RZ_2)$. It was further pointed out in \cite{chan15mi} that $I(\RZ_V)$ is the minimum solution $`g$ to the \emph{residual independence relation}
\begin{align}
H(\RZ_V)-`g = \sum_{C\in \mcP} `1[H(\RZ_C)-`g`2]
\end{align}
for some $\mcP\in \Pi'(V)$. To get an intuition of the above relation, notice that $`g=0$ is a
solution when the joint entropy $H(\RZ_V)$ on the left is equal to the sum of entropies
$H(\RZ_C)$'s on the right for some partition $\mcP$. In other words, the MMI is the smallest
value of $`g$ removal of which leads to an independence relation, i.e., the total residual
randomness on the left is equal to the sum of individual residual randomness on the right
according to some partitioning of the random variables. It was further shown in \cite{chan15mi}
that there is a unique finest optimal partition to \eqref{eq:I} with a clustering interpretation
in \cite{chan16cluster}. The MMI is also computable in polynomial time, following the result of
Fujishige~\cite{fujishige88}.

In the opposite extreme with $R\to 0$, it is easy to argue that
\begin{align}
\CS(0) \geq J_{\op{GK}}(\RZ_V) \label{eq:CS0}
\end{align}
where $J_{\op {GK}}(\RZ_V)$ is the multivariate extension of the G\'acs--K\"orner common information in \eqref{eq:JGK:2}
\begin{align}
J_{\op {GK}}(\RZ_V) := \sup \Set {H(\RU) \mid H(\RU|\RZ_i)=0\;\forall i\in V}
\label{eq:JGK}
\end{align}
with $\RU$ again chosen as a discrete random variable.
Note that, even without any public discussion, every user can compress their source independently to $\RU^n$ where $\RU$ is the maximum common function if $J_{\op {GK}}(\RZ_V)$ is finite. Hence, it is easy to achieve a secret key rate of $H(\RU)=J_{\op {GK}}(\RZ_V)$ without any discussion. The reverse inequality of \eqref{eq:CS0} seems plausible but has not been proven yet except in the two-user case. The technique in \cite{csiszar00} which relies on the Csisz\'ar sum identity does not appear to extend to the multi-user case to give a matching converse.

\subsection{Hypergraphical sources}
\label{sec:hyp}

Stronger results have been derived for the following special source model:
\begin{definition}[Definition~2.4 of \cite{chan10md}]\label{def:hyp}
	$\RZ_V$ is a \emph{hypergraphical source} w.r.t.\  a hypergraph $(V,E,`x)$ with edge functions $`x: E\to2^V`/\{\emptyset\}$ iff, for some independent edge variables $\RX_e$ for $e\in E$ with $H(\RX_e)>0$,
	\begin{equation}
	\RZ_i:=(\RX_e\mid  e\in E, i\in`x(e))\kern1em \text{ for } i\in V. \label{eq:Xe}
	\end{equation}
	In the special case when the hypergraph is a graph, i.e., $\abs {`x(e)}=2$, the model reduces to the pairwise independent network (PIN) model in \cite{nitinawarat10}. The hypergrahical source can also be viewed as a special case of the finite linear source considered in~\cite{chan10phd} if the edge random variables take values from a finite field.
\end{definition}

For hypergraphical sources, various bounds on $\RS$ and $\CS(R)$ have been derived in \cite{mukherjee16,MKS16,chan16itw,chan17cska}. The achieving scheme makes use of the idea of decremental secret key agreement~\cite{chan16isit,chan17idska}, where the redundant or less useful edge variables are removed or reduced before public discussion. This is a special case of the compressed secret key agreement, where the compression step simply selects the more useful edge variables up to the joint entropy limit.

For the PIN model, it turns out that decremental secret key agreement is optimal, leading to a single-letter characterization of $\RS$ and $\CS(R)$ in \cite{chan17cska}:
\begin{subequations}
	\label{eq:PIN:CSR,RS}
	\begin{align}
	\RS &= (\abs {V}-2) \CS(`8).\label{eq:PIN:RS}\\
	\CS(R) &= \min\Set*{\frac {R}{ \abs {V}-2},\CS(`8)} \quad \text{for $R\geq 0$.} \label{eq:PIN:CSR}
	\end{align}
\end{subequations}
It can be verified that \eqref{eq:PIN:RS} is the smallest value of $R$ such that $\CS(R)=\CS(`8)$ using \eqref{eq:PIN:CSR}. While the proof of converse, i.e., $\leq$ for \eqref{eq:PIN:CSR}, is rather involved, the achievability is by a simple tree packing protocol, which belongs to the decremental secret key agreement approach that removes excess edges unused for the maximum tree packing. In other words, the achieving scheme is a compressed secret key agreement scheme. This connection will lead to a single-letter characterization of $\tCS(`a)$ for the PIN model (in Theorem~\ref{thm:PIN}).

To illustrate the above results, a single-letter characterization for $\CS(R)$ will be derived in the following for the source in Example~\ref{eg:mot}. It will also demonstrate how an exact characterization for $\CS(R)$ can be extended from a PIN model to a hypergraphical model via some contrived arguments. The characterization will also be useful later in Example~\ref{eg:mot:2} to give an exact characterization of $\tCS(`a)$.

\begin{example}
	\label{eg:mot:CSR}
	The source defined in \eqref{eq:mot} in Example~\ref{eg:mot}, for instance, is a hypergraphical source with $E=\Set {a,b,c}$, $`x(a)=\Set {1,2,3}$, $`x(b)=\Set {1,2}$ and $`x(c)=\Set {2,3}$.
	By \eqref{eq:RCO}, we have $\RCO = 1$ with the optimal solution $r_1=r_3=0$ and $r_2=1$. This means that user $2$ needs to discuss $1$ bit to attain omniscience. In particular, user $2$ can reveal the XOR $\RX_b\oplus \RX_c$ so that user $1$ and $3$ can recover $\RX_c$ and $\RX_b$ respectively from their observations. 
	By \eqref{eq:CSR:RCO}, then, we have
	\begin{align}
	\CS(R)=\CS(`8)=H(\RZ_V)-\RCO = 2\kern1em \text{for $R\geq \RCO=1$}.\label{eq:mot:CSR1}
	\end{align}
	It can also be checked that the alternative characterization of $\CS(`8)$ in \eqref{eq:mmi} gives
	\begin{align*}
	\CS(`8)=I(\RZ_V) = \frac12 `1[H(\RZ_1)+H(\RZ_2)+H(\RZ_3)-H(\RZ_{\Set {1,2,3}})`2] = 2.
	\end{align*}
	
	Next, we argue that 	
	\begin{align}
	\CS(R)= 1+R \kern1em \text{for $R\in [0,1]$.} \label{eq:mot:CSR2}
	\end{align}
	The achievability, i.e., the inequality $\CS(R)\geq 1+R$, is by the usual time-sharing argument.
	In particular, the bound $\CS(0.5)\geq 1.5$, for example, can be achieved by the compressed secret key agreement scheme in Example~\ref{eg:mot} with $`a=2$, i.e., by time-sharing the compressed secret key agreement schemes for $`a=1$ and for $`a=3$ equally. More precisely, we set $\tRZ_1=(\RX_a^n,\RX_b^{\lfloor n/2\rfloor})$, $\tRZ_2=(\RX_a^n,\RX_b^{\lfloor n/2\rfloor},\RX_c^{\lfloor n/2\rfloor})$, $\tRZ_3=(\RX_a^n,\RX_c^{\lfloor n/2\rfloor})$, $\RK=(\RX_a^n,\RX_b^{\lfloor n/2\rfloor})$ and $\RF=\RF_2= \RX_b^{\lfloor n/2\rfloor}\oplus \RX_c^{\lfloor n/2\rfloor}$. It follows that the public discussion rate is $\limsup_{n\to\infty}\frac 1n \log\abs {F}=0.5$.
	
	Now, to prove the reverse inequality $\leq$ for~\eqref{eq:mot:CSR2}, we modifies the source $\RZ_V$ to another source $\RZ'_V$ defined as follows with an additional uniformly random and independent bit $\RX_d$:
	\begin{align*}
	\RZ'_1:=(\RX_a,\RX_b),\kern1em \RZ'_2:=(\RX_a,\RX_b,\RX_c,\RX_d),\kern1em \text {and}\kern1em \RZ'_3:=(\RX_c,\RX_d).\label{eq:mot:PIN}
	\end{align*}
	N.b., $\RZ'_V$ is different from $\RZ_V$, namely, $\RZ'_2$ is obtained from $\RZ_2$ by adding
	$\RX_d$, and $\RZ'_3$ is obtained from $\RZ_3$ by adding $\RX_d$ and removing $\RX_a$.
	%in the sense that $\RX_a$ is removed from $\RZ'_3$ but $\RX_d$ is added to both $\RZ'_2$ and $\RZ'_3$.
	It follows that $\RZ'_V$ is a PIN. By \eqref{eq:mmi} and \eqref{eq:PIN:CSR}, the constrained secrecy capacity for the modified source $\RZ'_V$ is
	\begin{align*}
	\CS'(R) = \min\Set*{R,2}.
	\end{align*}
	The desired inequality is proved if we can show that
	\begin{align*}
	\CS'(R+1) \geq \CS(R).
	\end{align*}
	To argue this, note that, if user $2$ reveals $\RF_2'=\RX_a\oplus \RX_d$ in public, then user $3$ can recover $\RX_a$. Furthermore, $\RF_2'$ does not leak any information about $\RX_a$, and so the source $\RZ'_V$ effectively emulates the source $\RZ_V$. Consequently, any optimal discussion scheme $\RF_V$ that achieves $\CS(R)$ for $\RZ_V$ can be used to achieve the same secret key rate but after an additional bit of discussion $\RF_2'$. This gives the desired inequality that establishes \eqref{eq:mot:CSR2}.
\end{example}

\section{Multi-letter characterization}
\label{sec:ML}

We start with a simple multi-letter characterization of the compressed secrecy capacity in terms of the MMI~\eqref{eq:mmi}.
%It will be useful to consider also the following simpler multi-letter characterization of the compressed secrecy capacity.
\begin{proposition}
	\label{pro:tCS:ML}
	For any $`a\geq 0$, we have
	\begin{align}
	\tCS(`a) &= \sup \lim_{n\to\infty} \frac1n I(\tRZ_V)
	\label{eq:tCS:ML}
	\end{align}
	where the supremum is over all valid compressed source $\tRZ_V$ satisfying the joint entropy limit~\eqref{eq:`a}. 
\end{proposition}
\begin{proof}
	This is because the compressed secrecy capacity is simply the secret key agreement on a compressed source. Hence, by \eqref{eq:mmi}, the MMI on the compressed source gives the compressed secrecy capacity.\qed
\end{proof}
The characterization in \eqref{eq:tCS:ML} is simpler than the formulation in \eqref{eq:tCS} because it does not involve the random variables $\RF$ and $\RK$, nor the recoverability~\eqref{eq:recover} and secrecy~\eqref{eq:secrecy} constraints. Although such a multi-letter expression is not computable and therefore not accepted as a solution to the problem, it serves as an intermediate step that helps derive further results. More precisely, consider the bivariate case where $V=\Set {1,2}$. Then, \eqref{eq:tCS:ML} becomes
\begin{subequations}
	\label{eq:tCS2:ML}
	\begin{align}
	\tCS(`a) &= \sup \lim_{n\to\infty} \frac1n I(\tRZ_1\wedge \tRZ_2) \kern1em \text {where}\label{eq:tCS2:ML:1}\\
	&\limsup_{n\to `8} \frac1n H(\tRZ_1,\tRZ_2)-`a \leq 0
	\label{eq:tCS2:ML:2}
	\end{align}
	If in addition the joint entropy constraint~\eqref{eq:tCS2:ML:2} is replaced by the entropy constraint on user~$1$ only, i.e.,
	\begin{align}
	\limsup_{n\to `8} \frac1n H(\tRZ_1)-`a \leq 0, \label{eq:tCS2:ML:3}
	\end{align}
\end{subequations}
then $\tCS(`a)$ can be single-letterized by standard techniques as in \cite{csiszar00} to $\tdC_{\op {S},1}(`a)$ defined in \eqref{eq:tCS1}. The following gives a simple upper bound that is tight for sufficiently small $`a$.
\begin{proposition}
	\label{pro:UB1}
	$\tdC_{\op {S},1}(`a)$ defined in \eqref{eq:tCS1} is continuous, non-decreasing and concave in $`a\geq 0$ with
	\begin{align}
	\tdC_{\op {S},1}(`a)\leq `a. \label{eq:UB1}
	\end{align}
	Furthermore, equality holds iff $`a\leq J_{\op{GK}}(\RZ_1\wedge \RZ_2)$.
\end{proposition}
\begin{proof}
	Monotonicity is obvious. Continuity and concavity can be shown by the usual time-sharing argument as in Proposition~\ref{pro:basic}.	\eqref{eq:UB1} follows directly from the data processing inequality that $I(\RZ'_1 \wedge \RZ_2 ) \leq I(\RZ'_1 \wedge \RZ_1 )$ under the Markov chain $\RZ'_1"-"\RZ_1"-"\RZ_2$ required in \eqref{eq:tCS1:3}. If $`a\leq J_{\op {GK}}(\RZ_1\wedge \RZ_2)$, then there exist a feasible solution $\RU$ to \eqref{eq:JGK:2} (a common function of $\RZ_1$ and $\RZ_2$) with $H(\RU)\geq `a$, and so the compressed sources $\tRZ_1$ and $\tRZ_2$ can be chosen as a function of $\RU^n$ to achieve the equality for \eqref{eq:UB1}. Conversely, suppose $J_{\op {GK}}(\RZ_1\wedge \RZ_2)$ is finite and \eqref{eq:UB1} is satisfied with equality. Then, in addition to $\RZ'_1"-"\RZ_1"-"\RZ_2$, we also have $\RZ'_1"-"\RZ_2"-"\RZ_1$, which implies by the double Markov property that, for the maximum common function $\RU$ achieving $J_{\op{GK}}(\RZ_1\wedge \RZ_2)$ defined in \eqref{eq:JGK:2},
	\begin{align*}
	I(\RZ'_1\wedge \RZ_1,\RZ_2|\RU) = 0\kern1em \text {(or $\RZ'_1"-"\RU"-"(\RZ_1,\RZ_2)$)}. \label{eq:UB1:opt}
	\end{align*}
	In other words, the optimal $\RZ'_1$ is a stochastic function of the maximum common function of $\RZ_1$ and $\RZ_2$, and so $`a=I(\RZ'_1\wedge \RZ_2) \leq J_{\op{GK}}(\RZ_1\wedge \RZ_2)$ as desired.\qed
\end{proof}
We will show that the above upper bound in \eqref{eq:UB1} extends to the multi-user case (in Proposition~\ref{pro:GK}). However, for $`a\geq J_{\op {GK}}(\RZ_1\wedge \RZ_2)$, the above upper bound is not tight even in the two-user case. To improve the upper bound, the following duality between $\tdC_{\op {S},1}$ and $C_{\opS,1}$ will be used and extended to the multi-user case (in Theorem~\ref{thm:duality}).
\begin{proposition}
	\label{pro:duality1}
	For $`a\geq J_{\op {GK}}(\RZ_1\wedge \RZ_2)$,
	\begin{align}
	\tdC_{\op {S},1}(`a) = C_{\op {S},1}(`a-\tdC_{\op {S},1}(`a)).\label{eq:duality1}
	\end{align}
	Furthermore, the set of optimal solutions to the left (achieving $\tdC_{\op {S},1}(`a)$ defined in \eqref{eq:tCS1}) is the same as the set of optimal solutions to the right (achieving $C_{\op {S},1}(R)$ in \eqref{eq:CS1} with $R=`a-\tdC_{\op {S},1}(`a)$). It follows that the minimum admissible entropy~\eqref{eq:RS} but with the entropy constraint on user~$1$ instead is
	\begin{align}
	`a_{\opS,1} &= R_{\opS,1}+I(\RZ_1\wedge \RZ_2) = J_{\opW,1}(\RZ_1\wedge \RZ_2)
	\end{align}
	where $R_{\opS,1}$ and $J_{\opW,1}(\RZ_1\wedge \RZ_2)$ are defined in \eqref{eq:RS1} and \eqref{eq:JW1} respectively.
\end{proposition}

\begin{proof}
	Set $R=`a-\tdC_{\op {S},1}(`a)$. 
	Consider first an optimal solution $\RZ'_1$ to $\tdC_{\op {S},1}(`a)$ and show that it is also an optimal solution to $C_{\op {S},1}(R)$. By optimality,
	\begin{align}
	I(\RZ'_1\wedge \RZ_2)=\tdC_{\op {S},1}(`a). \label{eq:duality1:t1}
	\end{align}
	By the constraint \eqref{eq:tCS1:2}, $I(\RZ'_1\wedge \RZ_1)\leq `a$. It follows that the constraint \eqref{eq:CS1:2} holds, and so $\RZ'_1$ is a feasible solution to $C_{\op {S},1}(R)$, i.e., we have $\geq$ for \eqref{eq:duality1} that
	\begin{align}
	\tdC_{\op {S},1}(`a) \geq C_{\op {S},1}(`a-\tdC_{\op {S},1}(`a)).\label{eq:duality1:t2}
	\end{align}
	To show that $\RZ'_1$ is also optimal to $C_{\op {S},1}(R)$, suppose to the contrary that there exists a strictly better solution $\RZ''_1$ to $C_{\op {S},1}(R)$, i.e., with
	\begin{align}
	I(\RZ''_1\wedge \RZ_2)>I(\RZ'_1\wedge \RZ_2)=\tdC_{\op {S},1}(`a).\label{eq:duality1:t3}
	\end{align}	
	It follows that 
	\begin{align}
	I(\RZ''_1\wedge \RZ_1)>I(\RZ'\wedge \RZ_1)=`a.\label{eq:duality1:t4}
	\end{align}
	The last equality means that the constraint~\eqref{eq:tCS1:2} is satisfied with equality. If to the contrary that the equality does not hold, setting $\RZ'_1$ to be $\RZ''_1$ for some fraction $`l>0$ of time gives a better solution to  $C_{\op {S},1}(R)$, contradicting the optimality of $\RZ'_1$. The first inequality can also be argued similarly by the optimality of $\RZ'_1$. Now, we have
	\begin{align*}
	\frac {I(\RZ''_1\wedge \RZ_2)-I(\RZ'_1\wedge \RZ_2)}{I(\RZ''_1\wedge \RZ_1)-I(\RZ'_1\wedge \RZ_1)}\utag{a}\leq \frac{I(\RZ'_1\wedge \RZ_2)}{I(\RZ'_1\wedge \RZ_1)} \utag{b}\leq 1,
	\end{align*}
	where \uref{a} is by the concavity of $\tdC_{\op {S},1}(`a)$; and \uref{b} is by the upper bound $\tdC_{\op {S},1}(`a)\leq `a$ in~\eqref{eq:UB1}. N.b., equality cannot hold simultaneously for \uref{a} and \uref{b} because, otherwise, we have $\frac{I(\RZ''_1\wedge \RZ_2)}{I(\RZ''_1\wedge \RZ_1)}=1$, which, together with \eqref{eq:duality1:t3} and \eqref{eq:duality1:t4}, contradicts the result in Proposition~\ref{pro:UB1} that $\tdC_{\op {S},1}(`a)< `a$ (with strict inequality) for $`a> J_{\op {GK}}(\RZ_1\wedge \RZ_2)$. Hence,
	\begin{align*}
	\frac {I(\RZ''_1\wedge \RZ_2)-I(\RZ'_1\wedge \RZ_2)}{I(\RZ''_1\wedge \RZ_1)-I(\RZ'_1\wedge \RZ_1)}< 1,
	\end{align*}
	which, together with \eqref{eq:duality1:t3} and \eqref{eq:duality1:t4}, implies
	\begin{align*}
	I(\RZ''_1\wedge \RZ_1)-I(\RZ''_1\wedge \RZ_2) > `a-\tdC_{\op {S},1}(`a) =R
	\end{align*}
	contradicting even the feasibility of $\RZ''_1$ to $C_{\op {S},1}(R)$, namely, the constraint~\eqref{eq:CS1:2} with $\RZ'_1$ replaced with $\RZ''_1$. This completes the proof of the optimality of $\RZ'_1$ to $C_{\op {S},1}(R)$.
	
	Next, consider showing that an optimal solution $\RZ'_1$ to $C_{\op {S},1}(R)$ is also optimal to $\tdC_{\op {S},1}(`a)$. Then, 	
	\begin{align*}
	I(\RZ'_1\wedge \RZ_1)\leq R+I(\RZ'_1\wedge \RZ_2)=`a-\tdC_{\op {S},1}(`a)+C_{\op {S},1}(R) \leq `a
	\end{align*}
	where the first inequality is by \eqref{eq:CS1:2}; the second equality is by the optimality of $\RZ'_1$; and the last inequality follows from \eqref{eq:duality1:t2}. Hence, the constraint~\eqref{eq:tCS1:2} holds and so $\RZ'_1$ is a feasible solution for $\tdC_{\op {S},1}(`a)$. If to the contrary that we have a better solution $\RZ''_1$ for $\tdC_{\op {S},1}(`a)$, then $\RZ''_1$ can be shown to be a feasible solution for $C_{\op {S},1}(R)$, contradicting the optimality of $\RZ'_1$.\qed
\end{proof}
%\begin{proposition}
%	For $`a\leq J_{\op{GK}}(\RZ_1\wedge \RZ_2)$, $\tdC_{\op {S},1}=`a$. For $`a\geq H(\RZ_1^*)$, $\tdC_{\op {S},1}=I(\RZ_1\wedge \RZ_2)$. For other values of $`a$, we have $\tdC_{\op {S},1}(`a)=C_{\opS,1}(`a-\tdC_{\op {S},1}(`a))$. Furthermore, the set of optimal solution to $\tdC_{\op {S},1}(`a)$ is the same as the set of optimal solution to $C_{\opS,1}(`a-\tdC_{\op {S},1}(`a))$.
%\end{proposition}
%, defining for any $\RZ'_1$, $\RZ''_1$, and $`l\in (0,1)$,
%\begin{align}
%\RZ^{`l}_1 := \begin{cases}
%\RZ'_1 & \text {with probability $1-`l$}
%\RZ''_1 & \text {with probability $`l$}
%\end{cases}
%\end{align}

\section{Main Results}
\label{sec:results}

The following extends the single-letter upper bound~\eqref{eq:UB1} in Proposition~\ref{pro:UB1} to the muli-user case.
\begin{proposition}
	\label{pro:GK}
	$\tCS(`a)\leq `a$ with equality if $`a\leq J_{\op {GK}}(\RZ_V)$.
\end{proposition}

\begin{proof}%[Sketch for Proposition~\ref{pro:GK}]
	The upper bound $\tCS(`a)\leq `a$ is because $n\tCS(`a)$ cannot exceed the unconstrained secrecy capacity for the compressed source $\tRZ_V$, which, by \eqref{eq:CSRCO} and \eqref{eq:`a}, is upper bounded by $H(\tRZ_V)\leq n`1[`a+`d_n`2]$ for some $`d_n\to 0$ as $n\to `8$. 
	
	Next, to prove the equality condition is sufficient, suppose $`a\leq J_{\op {GK}}(\RZ_V)$. Then, each user can compress their source directly to a common secret key at rate $`a$ without any public discussion. Hence, $\tCS(`a)=`a$ as desired.\qed
\end{proof}
N.b., unlike the two-user case in Proposition~\ref{pro:UB1}, the equality condition above in terms of the multivariate G\'acs--K\"orner common information is sufficient but not shown to be necessary. Nevertheless, necessity seems very plausible, as there seems to be no counter-example that suggests otherwise.

%%%%%%%%%%%%%%%%%%%%%%%%%%%%%%%%%%%

As in Proposition~\ref{pro:duality1}, a duality can be proved in the multi-user case, relating the compressed secret key agreement problem to the constrained secrecy key agreement problem.
%It turns out that the upper bound can be improved by 
\begin{theorem}
	\label{thm:duality}
	With $\CS(R)$ and $\RS$ defined in \eqref{eq:CSR} and \eqref{eq:RS} respectively, we have
	\begin{subequations}
		\begin{align}
		`a_{\opS} &\geq \RS + \CS(`8)\label{eq:`aSRS}\\
		\tCS(`a) &\leq \CS(`a-\tCS(`a)) \label{eq:tCSRCSR}
		\end{align}
		\label{eq:duality}
	\end{subequations}
	for all $`a\geq 0$.
\end{theorem}

\begin{proof}
	\eqref{eq:`aSRS} can be obtained from \eqref{eq:tCSRCSR} by setting $`a=`a_{\opS}$ as follows:
	\begin{align*}
	\CS(`8) \utag{a}\geq \CS(`a_{\opS} - \tCS(`a_{\opS})) 
	\utag{b}\geq \tCS(`a_{\opS}) \utag{c}= \CS(`8)
	\end{align*}
	where \uref{b} is given by \eqref{eq:tCSRCSR} with $`a=`a_{\opS}$; while \uref{a} and \uref{c} follows directly from \eqref{eq:C`8}, \eqref{eq:C`8:`a} and monotonicity. It follows that the inequalities \uref{a} and \uref{b} hold with equality. In particular, equality for \uref{a} means that $\CS(R)=\CS(`8)$ for $R\geq `a_{\opS}-\tCS(`a_{\opS})=`a_{\opS}-\CS(`8)$, implying \eqref{eq:`aSRS} as desired.
	
	To show \eqref{eq:tCSRCSR}, consider an optimal compressed secret key agreement scheme achieving $\tCS(`a)$ with an arbitrary entropy limit $`a$. 
	It suffices to show that the discussion rate need not be larger than $`a-\tCS(`a)$.
	Let $\tRZ_V$ be the optimal compressed source and $\tRCO$ be the smallest rate of communication for omniscience of $\tRZ_V$, which is given by \eqref{eq:RCO} with $\RZ_V$ replaced by $\tRZ_V$. The  discussion rate for the omniscience strategy is
	\begin{align*}
	\frac1n \tRCO &= \frac1n `1[H(\tRZ_V) - I(\tRZ_V)`2]
	\end{align*}
	by \eqref{eq:CSRCO}. This simplifies to $`a-\tCS(`a)$ as desired in the limit $n\to `8$. N.b., since the omniscience strategy is non-interactive, the desired hold even if $\CS$ and $R_{\opS}$ are defined with non-interactive discussion.\qed
\end{proof}

While it is obvious from the above proof that a compressed secret key agreement scheme can be used as a constrained secret key agreement scheme, yielding one of the best lower bounds for $\CS(R)$ in \cite{chan17cska},
the above result also means that a converse result on constrained secret key agreement can be applied to compressed secret key agreement. Upper bounds on $\tCS(`a)$ may be obtained from the upper bounds for $\CS(R)$ such as those in \cite{chan17cska}. It turns out that this approach can give better upper bounds which, surprisingly, is tight for the PIN model as mentioned in Section~\ref{sec:hyp}. This leads to the following exact single-letter characterization of $\tCS(`a)$.

\begin{theorem}
	\label{thm:PIN}
	For the PIN model in Definition~\ref{def:hyp},
	\begin{subequations}
		\label{eq:PIN:tCS`aS}
		\begin{align}
		`a_{\opS} &= (\abs {V}-1) \CS(`8)\label{eq:PIN:`aS}\\
		\tCS(`a) &= \min \Set*{\frac{`a}{\abs {V}-1},\CS(`8)}\label{eq:PIN:tCS}
		\end{align}
		for all $`a\geq 0$.
	\end{subequations}
\end{theorem}

\begin{proof}
	\eqref{eq:PIN:`aS} follows easily from \eqref{eq:PIN:tCS} by setting the two terms in the minimization to be equal and solving for $`a$. To show \eqref{eq:PIN:tCS}, note that,		
	by \eqref{eq:PIN:CSR}, we have
	\begin{align*}
	\CS^{-1}(`g) = (\abs {V}-2) `g\kern1em \forall `g<\CS(`8)
	\end{align*}
	because $\CS(R)$ is non-decreasing and concave, and so it must be strictly non-decreasing before it reaches $\CS(`8)=\CS(`8)$.	Now, by \eqref{eq:tCSRCSR},
	\begin{align*}
	`a-\tCS(`a) &\geq \CS^{-1}(\tCS(`a))\\
	&= (\abs {V}-2) \tCS(`a)
	\end{align*}
	for any $`a\geq 0$ such that $\tCS(`a)<\CS(`8)$, i.e., for $`a\leq `a_S$, and so $\tCS(`a)\leq \frac{`a}{\abs {V}-1}$. This implies $\leq$ for \eqref{eq:PIN:tCS}. The bound is achievable by the same achieving scheme in~\cite[Theorem~4.4]{chan17cska} along the idea of decremental secrecy key agreement and the tree packing protocol in \cite{nitinawarat-ye10}. More precisely, every $(\abs {V}-1)$ bits of edge variable forming a spanning tree are turned into a secret key bit by the tree packing protocol. This results in the factor of $(\abs {V}-1)$ in \eqref{eq:PIN:tCS`aS}, which corresponds to the number of edges in a spanning tree.\qed
\end{proof}

%%%%%%%%%%%%%%%%%%%%%%%%%%%%%%%%%%%

For the more general source model, the idea of decremental secret key agreement needs to be refined because there need not be any edge variables to remove. The following is a simple extension that leads to a single-letter lower bound on $\tCS(`a)$.

\begin{theorem}
	\label{thm:LB}
	A single-letter lower bound on $\tCS(`a)$ is
	\begin{align}
	\tCS(`a) &\geq I(\RZ_V'|\RQ)\label{eq:LB}
	\end{align}
	for any random vector $(\RQ,\RZ_V')$ taking values from a finite set and satisfying
	\begin{subequations}
		\label{eq:LB:c}
		\begin{align}
		I(\RQ \wedge \RZ_V) &= 0 \label{eq:LB:1}\\
		H(\RZ_i'|\RZ_i,\RQ) &= 0 \kern1em \forall i\in V\label{eq:LB:2}\\
		H(\RZ_V'|\RQ) &\leq `a.\label{eq:LB:3}
		\end{align}
	\end{subequations}
	Furthermore, it is admissible to have $\abs {Q}\leq 3$.
\end{theorem}

\begin{proof}
	By \eqref{eq:LB:2}, we have $\RZ_i'=`x_i (\RZ_i,\RQ)$ for some function $`x_i$. W.l.o.g., let $Q:=\Set {1,\dots, k}$ for some integer $k>0$. Choose $\tRZ_i$ to be the following function of $\RZ_i^n$:
	\begin{align*}
	\tRZ_i &= ((`x_i(\RZ_{i`t},q)\mid n_{q-1} < `t \leq n_q ) \mid 1 \leq q\leq k)\kern1em \text {where}\\
	n_0 &= 0\kern1em \text {and}\kern1em n_q=`1\lfloor n\sum_{j=1}^q P_{\RQ}(j) `2\rfloor\kern1em  \text {for $1\leq q\leq k$.}
	\end{align*}
	Basically, $\RQ$ acts as a time-sharing random variable where $P_{\RQ}(q)$ is the fraction of time the source $\RZ_i$ is processed to $\RZ_i^{(q)}:=`x_i(\RZ_i,q)$, for $1\leq q\leq k$. More precisely, we have $\frac{n_q-n_{q-1}}{n}$ converge to $P_{\RQ}(q)$, and so
	\begin{align*}
	\frac1n I(\tRZ_V) &= \sum_{q=1}^k I(\RZ_V^{(q)}) \frac{n_q-n_{q-1}}{n}\\
	& "-{^{n\to `8}}>" I(\RZ_V'|\RQ).
	\end{align*}
	Similarly,
	\begin{align*}
	\frac1n H(\tRZ_V) &"-{^{n\to `8}}>" H(\RZ_V'|\RQ) \leq `a
	\end{align*}
	by \eqref{eq:LB:3}, satisfying the entropy limit \eqref{eq:`a}. Hence, $\tRZ_V$ is a valid compressed source, the unconstrained capacity of which is $I(\RZ_V'|\RQ)$, leading to the desired lower bound~\eqref{eq:LB}.
	
	The condition that $\abs {Q}\leq 3$ is admissible follows from the usual argument by the well-known Eggleston--Carath\'eodory theorem. More precisely, let
	\begin{align*}
	\rsfsS:=\{(I(\RZ_V'|\RQ=q),H(\RZ_V'|\RQ=q))\mid &P_{\RZ_V|\RQ=q}=P_{\RZ_V},\\ &\kern-1em H(\RZ_i'|\RZ_i,\RQ=q)=0\}.
	\end{align*}
	It can be seen that the conditions above are equivalent to \eqref{eq:LB:1} and \eqref{eq:LB:2} respectively, and so the set of feasible values to \eqref{eq:LB:c}, namely
	\begin{align*}
	(I(\RZ_V'|\RQ),H(\RZ_V'|\RQ)) = \sum_{q\in Q} P_{\RQ}(q) (I(\RZ_V^{q}),H(\RZ_V^{q})),
	\end{align*}
	is equal to the convex hull of $\rsfsS$. Since the dimension of $\rsfsS$ is at most $2$, the pair $(\tCS(\RZ_V),`a)$ can be obtained as a convex combination of at most $3$ points in $\rsfsS$ as desired by the Eggleston--Carath\'eodory theorem.\qed
\end{proof}

The main results above can be illustrated as follows using the hypergraphical source in Example~\ref{eg:mot} given earlier. In particular, an exact single-letter characterization of $\tCS(`a)$ will be derived, even though such an exact characterization is not known for general hypergraphical sources.

\begin{example}
	\label{eg:mot:2}
	Consider the source defined in \eqref{eq:mot} in Example~\ref{eg:mot}. It will be shown that \eqref{eq:mot:tCS1}  and \eqref{eq:mot:tCS2} are satisfied with equality, which gives the desired single-letter characterization of $\tCS(`a)$.
	
	It is easy to show that $J_{\op {GK}}(\RZ_V)=1$ since $\RX_a$ is the maximum common function of $\RZ_1$, $\RZ_2$ and $\RZ_3$. Hence, the reverse inequality of \eqref{eq:mot:tCS1} follows from Proposition~\ref{pro:GK}.
	
	The reverse inequality for \eqref{eq:mot:tCS2} can be argued using the bound in Theorem~\ref{thm:duality} by $\CS(R)$ and the characterization of $\CS(R)$ in Example~\ref{eg:mot:CSR}. More precisely, by \eqref{eq:mot:CSR1}, that the unconstrained secrecy capacity $\CS(`8)=2$. Then, by \eqref{eq:mot:CSR2}, we have
	$\CS^{-1}(`g)\leq `g-1$ for all $`g\leq \CS(`8)=2$. Now, by \eqref{eq:tCSRCSR}, 
	\begin{align*}
	`a - \tCS(`a) &\leq \CS^{-1}(\tCS(`a))\leq \tCS(`a)-1
	\end{align*}
	and so $\tCS(`a)\leq \frac{1+`a}{2}$ for $\tCS(`a)\leq 2$. This completes the proof.
\end{example}

\section{Conclusion and Extensions}
\label{sec:conclusion}

Inspired by the idea of decremental secret key agreement and its application to the constrained secret key agreement problem, we have formulated a multiterminal secret key agreement problem with a more general source compression step that applies beyond the hypergraphical source model. This formulation allow us to separate and compare the issues of source compression and discussion rate constraint in secret key agreement. While a single-letter characterization of the compressed secrecy capacity and admissible entropy limit remains unknown,  single-letter bounds have been derived and they are likely to be tight for the hypergraphical model, and possibly more general source models such as the finite linear source model~\cite{chan10phd}. For the PIN model, in particular, the bounds are tight, giving rise to a complete characterization of the capacity in Theorem~\ref{thm:PIN}. One way to improve the current converse results is to show whether the equality condition in Proposition~\ref{pro:GK} is necessary, that is, $\tCS(`a)<`a$ for $`a>J_{GK}(\RZ_V)$. By the duality in Theorem~\ref{thm:duality}, the condition is necessary if one can show that $\CS (0)= J_{\op {GK}}(\RZ_V)$, i.e., \eqref{eq:CS0} holds with equality. Such equality can be proved for hypergraphical as well as finite linear sources by extending the lamination techniques in \cite{chan17cska}. It is hopeful that a complete solution can be given for the finite linear source model and the well-known jointly gaussian source model. The bounds~\eqref{eq:duality} in the duality result may plausibly be tight for these special sources, in which case non-interactive discussion suffices to achieve the constrained secrecy capacity. The current achievability results may also be improved. In particular, for the two-user case with joint entropy constraint~\eqref{eq:tCS2:ML}, the lower bound in \eqref{eq:LB} can be improved to $\tCS(`a)\geq \max I(\RZ'_1\wedge \RZ'_2)$ where $I(\RZ'_1\wedge \RZ_1)+I(\RZ'_2\wedge \RZ_2)\leq `a$ and $\RZ'_1"-"\RZ_1"-"\RZ_2"-"\RZ'_2$. Whether this improvement is strict or is the best possible is not clear yet but an extension to the multi-user case seems possible. A related open problem is to characterize the $\CS(R)$ in the two-user case with two-way non-interactive discussion. A simpler question is whether two-way non-interactive discussion can be strictly better than one-way discussion.

As pointed out before, by regarding the secrecy capacity as a measure of mutual information, an optimal source compression scheme translates to a dimension reduction technique potentially useful for machine learning. 
A closely related line of work is the study of the strong data processing inequality in~\cite{erkip98,anantharam13,anantharam14}, in particular, the ratio $s^*(\RZ_1;\RZ_2):=\sup\frac{I(\RZ'_1 \wedge \RZ_2)}{I(\RZ_1' \wedge \RZ_1)}$ where, as in \eqref{eq:tCS1}, the supremum is taken over the choice of the conditional distribution $P_{\RZ'_1|\RZ_1,\RZ_2}$ such that $\RZ'_1"-"\RZ_1"-"\RZ_2$ forms a Markov chain and $I(\RZ'_1\wedge \RZ)>0$. It is straightforward to show that $\sup_{`a\geq 0}\frac{\tCS(`a)}{`a}$ for the two-user case in \eqref{eq:tCS2:ML} is upper bounded by $s^*(\RZ_1;\RZ_2)$ and $s^*(\RZ_2;\RZ_1)$. However, a sharper bound and a more precise mathematical connection may be possible, and the result may be extended to the multivariate case. Furthermore, the linearization considered in \cite{huang12} may potentially be adopted to provide a single-letter lower bound on the compressed secrecy capacity. As in \cite{anantharam13,LCV16}, the problem may also be related to a notion of maximum correlation appropriately extended to the multivariate case.

	\section*{Acknowledgment} 
	%%\addcontentsline{toc}{section}{Acknowledgment}
	%
	%%\input{ack}
	%
	
The author would like to thank Dr.\ Ali Al-Bashabsheh for pointing out a mistake in an earlier proof and Qiaoqiao Zhou for the discussion of the two-user case. The author would also like to thank Prof.\ Shao-Lun Huang, Prof.\ Navin Kashyap, and Manuj Mukherjee for their valuable comments and pointers to related work.

	\bibliographystyle{splncs03}
	\bibliography{IEEEabrv,ref}

\begin{thebibliography}{10}
\providecommand{\url}[1]{\texttt{#1}}
\providecommand{\urlprefix}{URL }

\bibitem{ahlswede93}
Ahlswede, R., Csisz{\'{a}}r, I.: Common randomness in information theory and
  cryptography---{P}art {I}: Secret sharing. {IEEE} Trans. Inf. Theory  39(4),
  1121--1132 (Jul 1993)

\bibitem{anantharam13}
Anantharam, V., Gohari, A., Kamath, S., Nair, C.: On maximal correlation,
  hypercontractivity, and the data processing inequality studied by erkip and
  cover. arXiv preprint arXiv:1304.6133  (2013)

\bibitem{anantharam14}
Anantharam, V., Gohari, A., Kamath, S., Nair, C.: On hypercontractivity and a
  data processing inequality. In: IEEE International Symposium on Information
  Theory Proceedings (ISIT). pp. 3022--3026. IEEE (2014)

\bibitem{bennett1988privacy}
Bennett, C.H., Brassard, G., Robert, J.M.: Privacy amplification by public
  discussion. SIAM journal on Computing  17(2),  210--229 (1988)

\bibitem{chan17idska}
Chan, C., Al-Bashabsheh, A., Zhou, Q.: Change of multivariate mutual
  information: from local to global, accepted by {IEEE} Trans. Inf. Theory in
  Aug, 2017

\bibitem{chan16isit}
Chan, C., Al-Bashabsheh, A., Zhou, Q.: Incremental and decremental secret key
  agreement. In: IEEE International Symposium on Information Theory Proceedings
  (ISIT). pp. 2514--2518 (July 2016)

\bibitem{chan16so}
Chan, C., Al-Bashabsheh, A., Zhou, Q., Ding, N., Liu, T., Sprintson, A.:
  Successive omniscience. {IEEE} Trans. Inf. Theory  62(6),  3270--3289 (June
  2016)

\bibitem{chan16cluster}
Chan, C., Al-Bashabsheh, A., Zhou, Q., Kaced, T., Liu, T.: Info-clustering: A
  mathematical theory for data clustering. IEEE Transactions on Molecular,
  Biological and Multi-Scale Communications  2(1),  64--91 (June 2016)

\bibitem{chan16itw}
Chan, C., Mukherjee, M., Kashyap, N., Zhou, Q.: When is omniscience a
  rate-optimal strategy for achieving secret key capacity? In: IEEE Information
  Theory Workshop (ITW). pp. 354--358 (Sep 2016)

\bibitem{chan17cska}
Chan, C., Mukherjee, M., Kashyap, N., Zhou, Q.: Secret key agreement under
  discussion rate constraints. In: IEEE International Symposium on Information
  Theory Proceedings (ISIT). pp. 1519--1523 (Jun 2017)

\bibitem{chan2008tightness}
Chan, C.: On tightness of mutual dependence upperbound for secret-key capacity
  of multiple terminals. arXiv preprint arXiv:0805.3200  (2008)

\bibitem{chan10phd}
Chan, C.: Generating secret in a network. Ph.D. thesis, Massachusetts Institute
  of Technology (2010)

\bibitem{chan11isit}
Chan, C.: The hidden flow of information. In: IEEE International Symposium on
  Information Theory Proceedings (ISIT). St. Petersburg, Russia (Jul 2011)

\bibitem{chan12ud}
Chan, C.: Matroidal undirected network. In: IEEE International Symposium on
  Information Theory Proceedings (ISIT). pp. 1498--1502 (July 2012)

\bibitem{chan15mi}
Chan, C., Al-Babsheh, A., Ebrahimi, J., Kaced, T., Liu, T.: Multivariate mutual
  information inspired by secret-key agreement. Proceedings of the IEEE
  103(10),  1883--1913 (Oct 2015)

\bibitem{chan16allerton}
Chan, C., Al{-}Bashabsheh, A., Zhou, Q., Liu, T.: Duality between feature
  selection and data clustering. In: 54th Annual Allerton Conference on
  Communication, Control, and Computing, Allerton Retreat Center, Monticello,
  Illinois (2016)

\bibitem{chan17oo}
Chan, C., Mukherjee, M., Kashyap, N., Zhou, Q.: On the optimality of secret key
  agreement via omniscience. CoRR  abs/1702.07429 (2017),
  \url{http://arxiv.org/abs/1702.07429}

\bibitem{chan10md}
Chan, C., Zheng, L.: Mutual dependence for secret key agreement. In:
  Proceedings of 44th Annual Conference on Information Sciences and Systems
  (2010)

\bibitem{courtade16}
Courtade, T.A., Halford, T.R.: Coded cooperative data exchange for a secret
  key. {IEEE} Trans. Inf. Theory  62(7),  3785--3795 (July 2016)

\bibitem{csiszar2008axiomatic}
Csisz{\'a}r, I.: Axiomatic characterizations of information measures. Entropy
  10(3),  261--273 (2008)

\bibitem{csiszar00}
Csisz{\'a}r, I., Narayan, P.: Common randomness and secret key generation with
  a helper. {IEEE} Trans. Inf. Theory  46(2),  344--366 (2000)

\bibitem{csiszar04}
Csisz{\'{a}}r, I., Narayan, P.: Secrecy capacities for multiple terminals.
  {IEEE} Trans. Inf. Theory  50(12),  3047--3061 (Dec 2004)

\bibitem{csiszar08}
Csisz{\'{a}}r, I., Narayan, P.: Secrecy capacities for multiterminal channel
  models. {IEEE} Trans. Inf. Theory  54(6),  2437--2452 (June 2008)

\bibitem{erkip98}
Erkip, E., Cover, T.M.: The efficiency of investment information. {IEEE} Trans.
  Inf. Theory  44(3),  1026--1040 (1998)

\bibitem{friedman01}
Friedman, N., Mosenzon, O., Slonim, N., Tishby, N.: Multivariate information
  bottleneck. In: Proceedings of the Seventeenth conference on Uncertainty in
  artificial intelligence. pp. 152--161. Morgan Kaufmann Publishers Inc. (2001)

\bibitem{fujishige88}
Fujishige, S.: Optimization over the polyhedron determined by a submodular
  function on a co-intersecting family. Mathematical Programming  42(1-3),
  565--577 (1988)

\bibitem{gac72}
G{\'{a}}cs, P., K{\"{o}}rner, J.: Common information is far less than mutual
  information. Problems of Control and Information Theory  2(2),  149--162 (Feb
  1972)

\bibitem{amin10a}
Gohari, A., Anantharam, V.: Information-theoretic key agreement of multiple
  terminals---{P}art {I}. {IEEE} Trans. Inf. Theory  56(8),  3973 --3996 (Aug
  2010)

\bibitem{huang12}
Huang, S.L., Zheng, L.: Linear information coupling problems. In: IEEE
  International Symposium on Information Theory Proceedings (ISIT). pp.
  1029--1033. IEEE (2012)

\bibitem{kaspi85}
Kaspi, A.: Two-way source coding with a fidelity criterion. IEEE Transactions
  on Information Theory  31(6),  735--740 (November 1985)

\bibitem{liu16}
Liu, J., Cuff, P., Verdú, S.: Key capacity for product sources with
  application to stationary gaussian processes. {IEEE} Trans. Inf. Theory
  62(2),  984--1005 (Feb 2016)

\bibitem{LCV16}
Liu, J., Cuff, P.W., Verd{\'{u}}, S.: Common randomness and key generation with
  limited interaction. CoRR  abs/1601.00899 (2016)

\bibitem{ma12}
Ma, N., Ishwar, P., Gupta, P.: Interactive source coding for function
  computation in collocated networks. IEEE Transactions on Information Theory
  58(7),  4289--4305 (July 2012)

\bibitem{maurer93}
Maurer, U.M.: Secret key agreement by public discussion from common
  information. {IEEE} Trans. Inf. Theory  39(3),  733--742 (1993)

\bibitem{mukherjee16}
Mukherjee, M., Chan, C., Kashyap, N., Zhou, Q.: Bounds on the communication
  rate needed to achieve {SK} capacity in the hypergraphical source model. In:
  IEEE International Symposium on Information Theory Proceedings (ISIT). pp.
  2504--2508 (July 2016)

\bibitem{mukherjee14}
Mukherjee, M., Kashyap, N., Sankarasubramaniam, Y.: Achieving {SK} capacity in
  the source model: When must all terminals talk? In: IEEE International
  Symposium on Information Theory Proceedings (ISIT). pp. 1156--1160 (June
  2014)

\bibitem{MKS16}
Mukherjee, M., Kashyap, N., Sankarasubramaniam, Y.: On the public communication
  needed to achieve sk capacity in the multiterminal source model. {IEEE}
  Trans. Inf. Theory  62(7),  3811--3830 (July 2016)

\bibitem{nitinawarat12}
Nitinawarat, S., Narayan, P.: Secret key generation for correlated gaussian
  sources. {IEEE} Trans. Inf. Theory  58(6),  3373--3391 (June 2012)

\bibitem{nitinawarat-ye10}
Nitinawarat, S., Ye, C., Barg, A., Narayan, P., Reznik, A.: Secret key
  generation for a pairwise independent network model. {IEEE} Trans. Inf.
  Theory  56(12),  6482--6489 (Dec 2010)

\bibitem{nitinawarat10}
Nitinawarat, S., Narayan, P.: Perfect omniscience, perfect secrecy, and steiner
  tree packing. {IEEE} Trans. Inf. Theory  56(12),  6490--6500 (Dec 2010)

\bibitem{shannon48}
Shannon, C.E.: A mathematical theory of communication. The Bell System
  Technical Journal  27(3),  379--423 (July 1948)

\bibitem{shwartz-ziv17}
Shwartz{-}Ziv, R., Tishby, N.: Opening the black box of deep neural networks
  via information. CoRR  abs/1703.00810 (2017),
  \url{http://arxiv.org/abs/1703.00810}

\bibitem{tishby15}
Tishby, N., Zaslavsky, N.: Deep learning and the information bottleneck
  principle. In: IEEE Information Theory Workshop (ITW). pp. 1--5 (April 2015)

\bibitem{tishby00}
Tishby, N., Pereira, F.C., Bialek, W.: The information bottleneck method. In:
  Thirty-Seventh Annual Allerton Conference on Communication, Control, and
  Computing. Allerton Retreat Center, Monticello, Illinois (Sep 1999)

\bibitem{tyagi13}
Tyagi, H.: Common information and secret key capacity. {IEEE} Trans. Inf.
  Theory  59(9),  5627--5640 (Sep 2013)

\bibitem{vatedka16}
Vatedka, S., Kashyap, N.: A lattice coding scheme for secret key generation
  from gaussian markov tree sources. In: IEEE International Symposium on
  Information Theory Proceedings (ISIT). pp. 515--519 (July 2016)

\bibitem{watanabe60}
Watanabe, S.: Information theoretical analysis of multivariate correlation. IBM
  Journal of Research and Development  4(1),  66--82 (1960)

\bibitem{watanabe10}
Watanabe, S., Oohama, Y.: Secret key agreement from correlated gaussian sources
  by rate limited public communication. IEICE transactions on fundamentals of
  electronics, communications and computer sciences  93(11),  1976--1983 (2010)

\bibitem{watanabe11}
Watanabe, S., Oohama, Y.: Secret key agreement from vector gaussian sources by
  rate limited public communication. IEEE transactions on information forensics
  and security  6(3),  541--550 (2011)

\bibitem{zhang15}
Zhang, H., Liang, Y., Lai, L.: Secret key capacity: Talk or keep silent? In:
  IEEE International Symposium on Information Theory Proceedings (ISIT). pp.
  291--295 (June 2015)

\end{thebibliography}
	
\end{document}